\documentclass[conference]{IEEEtran}
\usepackage{graphicx}

\usepackage{multirow}
\usepackage{textcomp}
  
\usepackage{amssymb}

\usepackage{amsthm}

\usepackage[cmex10]{amsmath}

\newcommand{\url}[1]{#1}
\usepackage{tikz}
\usetikzlibrary{arrows,decorations,positioning,backgrounds,fit,shapes,matrix,shapes.callouts,decorations.text,automata,decorations.pathreplacing}
\usepackage{ifthen}

\newcommand{\card}[1]{\lvert #1\rvert}
\def\E{\mathbb{E}}

\def\d{\mathrm{d}}
\def\p{\mathrm{p}}
\def\r{\mathrm{r}}

\def\Pc{\mathcal{P}}
\def\epl{\mathrm{EPL}}
\def\mepl{\mathrm{MEPL}}

\newcommand{\rmnum}[1]{\romannumeral #1}

\newcommand{\ignore}[1]{}

\newtheorem{theorem}{Theorem}
\newtheorem{lemma}{Lemma}

\newtheorem{corollary}{Corollary}

\newtheorem{definition}{Definition}
\newtheorem*{theorem*}{Theorem}

\begin{document}

%\begin{titlepage}

\title{Self-Adjusting Networks \\ to Minimize Expected Path Length}

%\numberofauthors{4}

\author{\IEEEauthorblockN{Chen Avin\IEEEauthorrefmark{1}, Michael Borokhovich\IEEEauthorrefmark{1}, Bernhard Haeupler\IEEEauthorrefmark{2} and Zvi Lotker\IEEEauthorrefmark{1}}
\IEEEauthorblockA{\IEEEauthorrefmark{1}Ben-Gurion University of the Negev\\
Emails: \{avin, borokhom, zvilo\}@cse.bgu.ac.il}
\IEEEauthorblockA{\IEEEauthorrefmark{2}Massachusetts Institute of Technology\\
Email: haeupler@mit.edu}
}

%\date{}


\maketitle

\begin{abstract}
Given a network infrastructure (e.g., data-center or on-chip-network) and a distribution on the source-destination requests, the expected path (route) length is an important measure for the performance, efficiency and power consumption of the network.
In this work we initiate a study on \emph{self-adjusting networks}: networks that use local-distributed mechanisms to adjust the position of the nodes (e.g., virtual machines) %(or processes) 
in the network to best fit the route requests distribution. 
Finding the optimal placement of nodes is defined as the minimum expected path length (MEPL) problem. 
%When the network is a line and computation is done centrally, it is known as the minimum linear arrangement (MLA) problem
%which is  to be NP-hard
This is a generalization of the minimum linear arrangement (MLA) problem where the network infrastructure is a line and the computation is done centrally. %, nevertheless,  MLA is known to be NP-hard.
In contrast to previous work, we study the distributed version and give efficient and simple approximation algorithms for interesting and practically relevant special cases of the problem. In particular, we consider grid networks in which the distribution of requests is a symmetric product distribution. In this setting, we show that a simple greedy policy of position switching between neighboring nodes to locally minimize an objective function, achieves good approximation ratios. We are able to prove this result using the useful notions of expected rank of the distribution and the expected distance to the center of the graph.
%We study (grid) networks that dynamically reorganize in order to minimize the expected routing delay and network load. Similar to self-adjusting data structures our networks try to adjust to any distribution of route requests. We show that simple greedy switches that try to locally minimize the objective function achieve good approximation ratios to the optimum for grid networks.
\end{abstract}

%\bigskip

%{\bf Keywords}: self-adjusting networks, communication networks, distributed algorithms, routing, approximation, minimum-linear-arrangement, energy-efficient

%\bigskip

%{\bf Regular presentation}

%\end{titlepage}

%\newpage
%\setcounter{page}{1}

\section{Introduction}\label{sec:intro}

% Motivation: Network types and applications, Goal and important metrics, Traffic statistics, local switches

% Self-Adjasting data structures

% Model: minimum expected path length problem

% Contribution: hardness, bound on best and local minimas configurations, local switches rules

% Future directions: p2p networks, self routing.

In the last decade we have witnessed two new major and related phenomena in distributed computing. The first is the emerge of huge data centers and warehouse-scale computers. The second phenomenon is the decentralization and parallelism of workload in single multi-core computers. In both cases (but on different scale) the system is a network of computing primitives that share global computational goals. In data centers networks, as well as in modern multiprocessor computers, multiple processes run in parallel to execute common tasks so, in many cases, these processes need to communicate with each other to work on their shared tasks.

Reducing energy waste, and in particular the power consumption of computing is one of the major challenges of the 21st century.
Both data centers and single computers are no exception, and constantly increasing their energy and power usage. For example, the total cost of power consumption of data centers in the USA alone is estimated to be 50 billion dollars~\cite{poess2008energy}. Moreover, the energy consumed by data centers is estimated to double every five years~\cite{energy2007report}. The focus of this work is to improve upon the energy that is consumed by routing in such systems. It is estimated that in data centers the energy consumed by routing is about 20\%-30\% of the total energy~\cite{heller2010elastictree}. Routing in network-on-chip (NoC) consumes even up to 50\% of the total energy~\cite{mirza2007empirical}. These numbers pose our community both an opportunity and a challenge. The opportunity is to gain significant energy savings for these systems; the challenge is to design and implement clever and simple algorithms that can improve routing efficiency.

Another common property of these systems is that they all operate in a fixed network infrastructure. This means that we cannot change the structure of the network by, for example, rewiring links. But instead, what we can do, is to move the locations of processes (e.g., virtual machines) between the different computers (or CPUs). 
%
%Considering the Software Defined Networks (SDN) concepts \cite{McKeown:2008:Openflow}, that promise much better control over the network, migration of processes will become a simple and common function that a network may support. A management software will have an abstraction for moving a selected process from one physical location to another. 
%
In this paper, we formulate the problem of saving energy on a fixed infrastructure network using migration of processes. 
The basic idea is that the energy cost of routing in a network is proportional to the length of the routes which suggests the following: If we can make the routes lengths (or the \emph{expected} route length) shorter, then we can save energy. We devise local and distributed algorithms that (re-)place processes in the network to reduce the expected path length.
This can be achieved, for example, by Software Defined Networking (SDN) \cite{SDN:2009} -- the concept, which provides, among others, much better control over the network functionality. 
In SDN, a software management platform may support an abstraction for moving a selected process from one physical machine to another. Recently, this approach became practical, when Google announced \cite{summit:google:talk} the implementation of OpenFlow \cite{McKeown:2008:Openflow} in its own backbone.

The problem of minimizing the total energy consumed by routing is dependent on two major properties of the system: (i) the infrastructure (topology) of the communication network and (ii) the statistical pattern of route requests between sources and destinations. We first show that even in a very simple pattern such as every node has an activity level and the probability to send or receive a message is proportional to its level, the problem is NP-complete on general network topologies. Secondly, even when the network is ``simpler'' or regular, like a grid network, the problem can still be hard if the request distribution is ``complex'' in some sense. With this in mind we turn to analyze approximation algorithms for the setting where both the topology and the requests have nice properties. Our routing and activity distributions are partially justified from real data~\cite{gummadi2003measurement,klemm2004characterizing}. We concentrate on local and distributed algorithms, namely, processes can be exchanged (i.e., relocated) only between nearby nodes without any centralized coordination.

\subsection{Overview of our results}

%\begin{itemize}
%\item Hardness.
%\item Line
%\item Grid
%\item Clusters
%\end{itemize}

First, we formulate the discussed problem as the \emph{minimum expected path length} (MEPL) problem, that is, given a network infrastructure and a distribution of requests, minimizing route costs by finding an optimal \emph{placement} for processes in the network.
When the network is a line, MEPL is identical to the minimum linear arrangement (MLA)~\cite{johnson1979computers} which is known to be NP-complete.
In this work we consider $d$-dimensional grid networks, $d\ge1$, and requests that comes independently from a \emph{symmetric product distribution} where the frequency of a route request $(u,v)$ is a multiplication of the \emph{activity} levels of both $u$ and $v$. 
In contrast to previous works, our goal is to design simple distributed algorithms for these more realistic settings.

We first show that MEPL is NP-complete if (\rmnum{1}): we only assume that the network is a $2$-dimensional grid, and (\rmnum{2}): we only assume that the requests come from a \emph{symmetric product distribution}.
But, somewhat surprisingly, if both conditions hold, we are able to present a simple, local, distributed algorithm that achieves good approximation to the optimal solution for the MEPL problem.
Our algorithm is self-adjustable in the sense that nodes switch processes based on the continuously observed sequence of route requests each node is involved in. This approach was inspired and bears some similarity to self-adjusting data structures like splay trees~\cite{sleator1985self}. In particular we are able to show (informal):
\begin{theorem*} For a \emph{$d$-dimensional grid} network and a \emph{symmetric product distribution} of requests there is a simple distributed algorithm, which defines a local switching policy between a process and its neighbors that achieves a \emph{constant} approximation to the minimum expected path length (MEPL) problem.
\end{theorem*}
Interestingly, we prove this theorem using a measure called \emph{expected rank} which is related to the uncertainty of a random variable in a similar manner as entropy is.

We then turn to more complex distributions of requests and discuss requests that are \emph{clustered} into disjoint groups. %In such a setting, our local algorithms are not guaranteed to always preform well.
While for this setting few extremely unstable bad local minima can exist we present promising simulation results. In particular we show that for the $2$-dimensional grid that starting from a random and thus almost worst case initial state of processes locations in the network our local algorithms converge to an almost optimal local minimum. 
% Nevertheless, we present promising simulation results that consider random initial states of processes locations in the network. For the $2$-dimensional grid, our local algorithms seems to perform very well.\\

%We take the stand that this study is only at its initial state. In this work, we define the problem and make first steps toward solving it, but 

\noindent {\bf Organization:}
In Section~\ref{sec:related} we discuss related work and somewhat similar approaches. Section~\ref{sec:model} introduces the formal problem and definitions. The hardness of MEPL is proved in Section~\ref{sec:hardness} and then in Section~\ref{sec:product} we prove our main result, a constant factor approximation in $d$-dimensional meshes with product route distributions.
In Section~\ref{sec:cluster} we discuss a more complex setting: clustered requests; and we end the paper with a short conclusion in Section~\ref{sec:conclusion}.
%Due to space constraints, some of the proofs are omitted and can be found in the full version \cite{Full_version_Archive}.
%, the details of this case are deferred to the appendix.
% \textbf{(Comment (Michael): we don't have discussion or conclusion section...)}
% and we conclude with discussion and future work in Section~\ref{sec:future}. 

\section{Related Work} \label{sec:related}

% self-adjusted traffic awareness protocols that makes use of instantaneous traffic information
% Gradient mechanism in a communication network
%Trafﬁc Engineering for grids for enabling grid networks to self-adjust to resource availability 

Energy saving along with green computing is an active topic of research in the recent years.
In a recent paper~\cite{lis2011brief} Lis et al. study memory architectures of microprocessors. The authors suggest that processes will migrate to a location that is closer to the data instead of what is common in today architectures, i.e., coping the data to be closer to the process. The logic behind this idea is that programs are much smaller than their data. We take this idea one step further by reducing the communication distance between two processes. 
Improving the energy efficiency of routing in networks was also considered. Batista et al.~\cite{batista2007self} used traffic engineering on grids to self-adjust to routing requests. 
In \cite{heller2010elastictree} and \cite{Shang2010Energy-aware} different authors considered data centers and tried to save energy by powers down routers and links when demand in the network  is low. Other self-adjusting routing scheme were considered, for example in scale-free networks to overcome congestion~\cite{PhysRevE.80.026114, Zhang2007177}.

The most related areas of research to our study are graph \emph{arrangement}, \emph{embedding} and \emph{labeling} problems \cite{chung1988labelings,Diaz2002A-survey}. The basic question there is to embed a guest graph $G$ into a host graph $H$ in order to minimize some objective function like the \emph{bandwidth} or the \emph{cutwidth}; we relate our study to this settings in the model section. In particular, some VLSI design problems where considered on a two dimensional grids.\cite{bhatt1987complexity,bhatt1984framework}. There are two significant differences here: first, we consider \emph{distributions} on the route requests which restrict our guest graphs and second, and more importantly, we are interested in distributed, self-adjusting algorithms to solve the problem and not a centralized solution.

As described in the introduction the self-adjusting nature of our solution was inspired by self-adjusting data structures like splay trees~\cite{sleator1985self} which adjust their structure according to requests made to the data structure in such a way that the amortized cost matches the cost of the optimal (static) solution. The local greedy switch strategy we use is related to physics and natural dynamics which indirectly try to minimize energy. Using this analogy for optimization purposes has a long history. E.g., simulated annealing~\cite{kirkpatrick1983optimization} can be seen as simulating physics while cooling the temperature, i.e., the local moves selected shift over time more and more bias from mostly random behavior to greedy energy minimization. Here we only look at greedy steps. 
In a networking context similar approaches were used for load balancing via diffusive paradigms~\cite{rabani1998local} and for routing via gradient mechanisms~\cite{PhysRevE.77.036121}.

Another very related research is about self-stabilizing graphs~\cite{jacob2009distributed,jacob2009self}. The goal there is also to maintain some objective using local edge exchanges,
mostly in an overlay network. In a similar manner we would like to extend the current work to solve MEPL on overlay and peer-to-peer networks, using edge rewiring as well.

% A very similar approach is also the spring embedding algorithm from the graph drawing community. $\ldots$\\

% Related work
%1) diffusive paradigm for Load balancing~\cite{rabani1998local}
% 2) Trafﬁc Engineering for Grids~\cite{batista2007self}
%3) PRODUCT DISTRIBUTION , query behavior of peers in a peer-to-peer~\cite{gummadi2003measurement,klemm2004characterizing}
%4) Gradient Mechanism, 2−d communication network of regular nodes~\cite{PhysRevE.77.036121}
%5) routing according to distribution in scale free~\cite{Zhang2007177}
% 6) energy-aware routing 
%7) self*

%Motivation: Data centers, saves energy, increase utilization. multiple CPUs, virtual machines.
%Future P2P networks.

%The local greedy switch strategy is related to physics and natural dynamics which indirectly try to minimize energy. Using this analogy for optimization purposes has a long history. E.g., simulated annealing can be seen as simulating physics while cooling the temperature, i.e., the local moves selected shift over time more and more bias from mostly random behavior to greedy energy minimization. Here we only look at greedy steps.\\

%A very similar approach is also the spring embedding algorithm from the graph drawing community. $\ldots$\\

%Self-adjusting data structures: Splay-Trees, $\ldots$\\

%Choosing the best tree to minimize expected length with respect to a distribution is exactly what a Huffman tree does when there is only one sink and only leaves carry data, i.e., and all the route requests are between the root and the leaves of a binary tree. $\ldots$\\

\section{Model and Problem Definition} \label{sec:model}
We model the communication network by an undirected, unweighted and connected \emph{host} graph $H$. 
Given a graph, its vertex set is denoted $V(H)$ and its edge set by $E(H)$. We denote the number of nodes with $n = \card{V(H)}$ and the number of edges with $m = \card{E(H)}$. Let $\d_H(\cdot)$ be the distance function between nodes in $H$, i.e., for two nodes $u,v \in H$ we define $d_H(u,v)$ to be the number of edges of a \emph{shortest} path between $u$ and $v$ in $H$. 

We assume that the network serves route requests drawn independently from an arbitrary distribution $\Pc$ and messages are routed along the shortest paths in $H$. 
Alternately, we represent the distribution $\Pc$ as a weighted directed \emph{guest} graph $G$ where $\card{V(G)} = n$.
For a directed edge $(u,v) \in E(G)$ let the weight of the edge $\p(u,v)$ denote the probability of a route request for a message
from node $u$ to $v$.

%Let  $M \in V \times V$ be a random variable that denote a route request for a message from node $u$ to $v$. Let $\Pc$ be a probability distribution on $M$, namely on the routes requests and for every $u,v \in V$, let $\p(u,v)$ denote the probability of the request $\{u,v \}$. % and recall that $d(\ru{u}{v})$ is the length of a shortest path between $u$ and $v$. 
%We can now define the \emph{expected path length}:
%
%\begin{definition}[\bf \emph{expected path length}]
%Given a network $G$ and a probability distribution $\Pc$ the expected path length of $G$ and $\Pc$ is defined as:
%\begin{align}
%\epl = \E_{G, \Pc}[\d(M)] = \sum_{u,v \in V} \p(u,v) \cdot d_G(u,v)
%\end{align}
%\end{definition}
%We drop the subscripts $G$ and/or $\Pc$ when they are clear from the context.
%Note that a special case of this definition is when $\Pc$ is the uniform distribution. This gives the \emph{average path length}
%which is often used in the literature instead of the diameter of the network, for example to show that a network is a \emph{small world network}~\cite{watts1998collective}.

Given a network infrastructure host graph $H$ and a distribution on the route requests represented by a guest graph $G$,
 a {\bf \emph{placement}} (or \emph{labeling} \cite{chung1988labelings}) function is
a bijective\footnote{In this work we consider the basic case where every host machine can run at most one virtual machine.} function $\varphi: V(G) \rightarrow V(H)$ which determines the locations of nodes of $G$ (processes) in the network $H$ (hosts).
%We write $\varphi(G)$ to denote a new graph $G' = (V,E')$ where $u,v$ is an edge in $G$ if and only if $\varphi(u),\varphi(v)$ is an edge in $G'$.
Given  $G, H$ and a placement  function $\varphi$ the expected path length of route requests is defined as:
$$
\epl(G, H, \varphi) = \sum_{uv \in E(G)} \mathrm{p}(u, v) \cdot \mathrm{d}_ {H}(\varphi(u),\varphi(v))
$$

When $H$ and $G$ are clear from the context we may write just $\epl(\varphi)$.
Note a special case of this definition, when $\Pc$ is the uniform distribution: this gives the \emph{average path length} in the network which is often used in the literature instead of the diameter, for example to show that a network is a \emph{small world network}~\cite{watts1998collective}.

For $H$ and $\Pc$ we would like to find an optimal placement of the nodes in the network to minimize the expected path length. Formally:
\begin{definition}[\bf \emph{Minimum Expected Path Length problem}]
Given a host graph $H$ and a probability distribution represented by a guest graph $G$, find a placement function that minimize the expected path length:
\begin{align}
\mepl = \min_\varphi \epl(G, H, \varphi)
\end{align}
\end{definition}

%we would like to find optimal placing of the nodes in the network so to minimize the expected path length

As mentioned earlier, this problem is motivated by the network serving point-to-point routing requests that are independently sampled from a distribution $\Pc$. If we assume that the cost for a request ${u,v}$ is $\d(u,v)$ then the MEPL problem tries to minimize the expected cost of a route. Note that this is also equivalent to minimizing the expected usage of all links or  minimizing the expected total number of transmissions - all important metrics in terms of energy saving and efficiency.

In this work, we mostly consider local and distributed switching rules to find a good placement: rules where a node is only allowed to switch places with nodes that are in its neighborhood (i.e., close to it). The goal is that after a sequence of local switches the network will reach its minimum expected path length and solve the MEPL problem.
On the one hand, our results from Section~\ref{sec:hardness} will show that this is not possible (efficiently) in a general setting even with global knowledge and non-local switches. Throughout the paper we thus consider at times simpler forms of networks and requests distributions, i.e., \emph{grid networks} and the \emph{symmetric product distributions}:

\begin{definition}[\bf \emph{$d$-dimensional grid networks}] A mesh network of size $n=k^d$ with nodes embedded on all locations 
$[k]^d$ where $[k]$ is the set of integers $1, 2, \dots, k$.
%\{$\Z_1^k \times \Z_1^k \times \dots \times \Z_1^k\}$ where $\Z_1^k$ is the integers $1, 2 \cdots, k$. 
Each node is connected to all the nodes at $\ell_1$-distance one from it, i.e., each node has at most two neighbors in each of the $d$ dimensions.
\end{definition}

\begin{definition}[\bf \emph{symmetric product distribution}]\label{dfn:indep_sym_distr}
In symmetric product distribution, each node has a level of activity and the more two nodes $u$ and $v$ are active the more likely that the route $\{u,v\}$ gets requested. More precisely, we scale the activity levels of the nodes such that they form a distribution with an activity level $\p(u)$ for each node $u$ and assume that the request distribution is induced by the product of the activity levels, i.e., $\p(u,v) = \p(u) \cdot \p(v)$.
\end{definition}

%We assume that the request are independently sampled from a distribution on $\binom{V}{2}$ and we denote the probability of a request ${u,v}$ being sampled with $p_{\{u,v\}}$. The objective of the self-adjustment strategy is to minimize the expected cost of a routing request, i.e., $\sum_{\{u,v\} \in \binom{V}{2}} p_{\{u,v\}} \d(u,v)$. Note that this is also the minimizing the expected usage of links (or expected total number of transmissions).

\section{Hardness of MEPL}\label{sec:hardness}

%Finding an optimal solution to the MEPL problem is hard in general. 
In this section we show that solving the general MEPL problem is hard. Indeed, we prove two results that demonstrate how the hardness of the problem can come form either an involved network topology $G$ or the structure of the routing request distribution $\Pc$. This serves also as an additional motivation why in the rest of the paper we turn to graphs and  distributions with more realistic structure.

For both our examples it suffices to use probability distributions that only have one non-zero probability value. In our first statement, we show that even if we restrict ourselves to symmetric product distributions, the MEPL problem is hard on general networks:

\begin{lemma}
Given a host graph $H$ and a symmetric product distribution $\Pc$, it is NP-complete to decide whether the MEPL is smaller than a given value.  
\end{lemma}
\begin{proof}
We describe a reduction from the $k$-CLIQUE problem. In the $k$-CLIQUE problem, one is given a graph $H'$ and has to decide whether $H'$ contains a $k$-clique, that is, whether $H'$ contains a complete graph on $k$ nodes as a subgraph. This is one of Karp's $21$ NP-complete problems~\cite{karp1972reducibility}. For the reduction we take the graph $H'$ as the network's host graph $H$. As a request distribution we use a symmetric product probability distribution that puts $1/k^{2}$ probability weight on each of the pairs $V' \times V'$ formed by a subset of the nodes of size $k$ and zero probability on any other pair. If $H$ contains a $k$-clique, then the unique optimal solution to MEPL with value $1 - 1/k$ will be obtained if all $k$ nodes are placed in this clique. If $H$ does not contain a $k$-clique the there will be at least one request pair $u,v \in V'\times V'$ that is at least two far apart and the total cost will be at least $1 - 1/k + 1/k^2$. Thus, deciding whether the MEPL is smaller than $1 - 1/k + 1/k^2$ is equivalent to deciding whether $H'$ has a $k$-CLIQUE and thus, NP-hard. Lastly, it is easy to see that deciding whether the MEPL is smaller than a given value problem can be easily achieved in NP by guessing and then verifying a solution with smaller value.  
\end{proof}

This lemma shows that solving the MEPL problem for general network topologies is hard. Next, we show that even if we restrict the graph to be nice, e.g., $2$-dimensional grid, a lack of structure in the probability distribution can make the MEPL problem hard, too:

\begin{lemma}\label{lem:2Dhard}\label{LEM:2DHARD}
Given a probability distribution $\Pc$, it is NP-complete to decide whether the MEPL is smaller than a given value on a $2$-dimensional grid network.  
\end{lemma}
\begin{proof}
%[Proof overview] 
%The proof is a simple reduction from the problem of embedding a tree with maximum degree of four in a $2$-dimensional grid, which was shown to be NP-complete by Bhatt and Cosmadakis~\cite{bhatt1987complexity}. We omit the details here.
We describe a reduction from the problem of embedding a tree in a $2$-dimensional grid which was shown to be NP-complete by Bhatt and Cosmadakis~\cite{bhatt1987complexity}. More precisely, it is NP-hard to decide whether a given tree $T$ (with maximum degree $4$) is a subgraph of the grid. Given an instance of this problem in form of a tree $T$ we construct a hard MEPL instance as follows: We take the two-dimensional grid $[k]^2$ as a topology where $k$ equals the number of nodes in $T$. As a request distribution we take a subset of $k$ nodes to correspond to nodes in $T$ and put a probability mass of $1/(k-1)$ on each pair of nodes that corresponds to two neighbors in the tree $T$; all other $k^2 - (k - 1)$ node pairs have a probability of zero. The MEPL for such an instance is $1$ if and only if the tree can be embedded in the grid. If this is not the case, then at least one request pair will be separated by a path of length at least two increasing the average to at least $1 + 1/(k-1)$. Thus deciding whether the MEPL is smaller than $1 + 1/(k-1)$ is equivalent to deciding whether $T$ can be embedded into the $2$-dimensional grid. This proves that solving the MEPL problem on the $2$-dimensional grid NP-hard. To show NP-completeness it is again easy to see that deciding whether the MEPL is smaller than a given value problem can be achieved in NP. 
\end{proof}

Contrasting these two hardness results, the next sections will show that if one assumes a grid graph \emph{and} a symmetric product distribution, nice algorithmic results can be obtained. 

\section{Distributed MEPL with Symmetric Product Distributions} \label{sec:product}

For general request distribution it is hard to find a good or optimal solution even when one is not restricted to local and distributed switching rules. With this in mind, we first restrict ourselves to a simpler model of requests, namely, symmetric product distributions (Definition~\ref{dfn:indep_sym_distr}). Second, we assume $d$-dimensional grid topologies, and in particular the line and a $2$-dimensional grid. We assume that a node learns the distribution from requests it is involved in and thus, it can decide whether the switching (exchanging positions) with a  neighbor will increase or decrease the objective function, the expected path length of the network. 
The main result of this section is that under the above settings, a good approximation to the objective function can be found using only simple (greedy) local switching rules. To prove this result, we need the following definitions.

\subsection{Expected Distance to Center and Expected Rank} 
To find a good placement  for nodes which gives a good approximation to the MEPL, we define the \emph{expected center} and the \emph{expected distance} to it.
\begin{definition}[\bf \emph{center and expected distance to the center}]
The expected center of a graph $H$ and a symmetric product distribution $\Pc$, is a node $c^*$, s.t.:
$$
c^* = \underset{x \in V(H)}{\operatorname{argmin}} \sum_{u \in V(H)}\p(u)\d(u,x).
$$
The expected distance to the center $c^*$ is:
$$
C=\min_{x \in V(H)} \sum_{u \in V(H)}\p(u)\d(u,x) \text{,} 
$$
or equally:
$$
C=\sum_{u \in V(H)}\p(u)\d(u,c^*).
$$
%or equally: 
%$$
%C=\sum_{u \in V}\p(u)\d(u,c^*)
%$$.
\end{definition}
When $H$ and $\Pc$ (represented by $G$) are clear in the context, both $C$ and $c^*$ can be written as functions of a \emph{placement } $\varphi$, i.e., $C(\varphi)$ and $c^*(\varphi)$. The minimum expected distance to the center is defined then as %$C_{\min}$:
$
C_{\min}=  \min_{\varphi} C(\varphi)
$.
The next lemma describes the relation between $C$ and EPL.
 \begin{lemma}\label{lem:epl-c}
For any given placement  $\varphi$,
$
2C(\varphi) \ge \epl(\varphi) \ge C(\varphi). 
$
\end{lemma}
\begin{proof}
To see the upper bound, we suppose that instead of routing between two nodes directly, we route every request via the center $c^*$. Routing a request in this way results in sampling two requests and summing up their distances to the center. In expectation, this is exactly $2C$. Formally (for any $\varphi$):
\begin{align*}
\epl(\varphi)&=\sum_{uv\in E(G)} \p(u) \p(v) \cdot \d(\varphi(u),\varphi(v))\\
&\le \sum_{uv\in E(G)} \p(u) \p(v) (\d(\varphi(u),c^*) +\d(c^*,\varphi(v))) \\
&=\sum_{u\in V(G)} \p(u)\d(\varphi(u),c^*)+\sum_{v\in V(G)} \p(v)\d(c^*,\varphi(v))\\
&= 2C(\varphi)
\end{align*}
The fact that $C$ is a lower bound, we show as follows:
\begin{align*}
\epl(\varphi)&=\sum_{uv\in E(G)} \p(u) \p(v) \cdot  \d(\varphi(u),\varphi(v))\\
&=\sum_{u\in V(G)} \p(u) \sum_{v\in V(G)}\p(v) \d(\varphi(u),\varphi(v))\\
&\geq \sum_{u\in V(G)} \p(u) \sum_{v\in V(G)}\p(v)  \d(c^*,\varphi(v))\\
&=\sum_{v\in V(G)}\p(v)  \d(c^*,\varphi(v))=C(\varphi)
\end{align*}
\end{proof}

\begin{corollary}
$\mepl \ge C_{\min}$
\end{corollary}
This follows since for the optimal placement $\varphi^*$: $\mepl = \epl(\varphi^*) \ge  C(\varphi^*) \ge C_{\min}$.

\noindent An important ingredient in bounding the performance of our local rules  will be the following measure of \emph{expected rank}.
This quantity is an interesting measure on the concentration and uncertainty of a distribution.
\begin{definition}[\bf \emph{Rank of nodes and the Expected rank}]\label{dfn:rank}
The rank of a node is the position of the node in the ordered list of nodes' probabilities (breaking ties arbitrarily). The node with the highest probability has rank $0$. The rank of the node $u \in V$ is denoted as $\r(u)$.
The expected rank of a probability distribution on the nodes of graph $G$ is:
$
\E\left[R\right]=\sum_{u\in V(G)}\p(u) \r(u)
$.
\end{definition}

We next describe the local switching rules by which our distributed algorithm works.

\subsection{(Greedy) Local Switching Strategies}
For two nodes $u,v \in V(G)$ and a placement $\varphi$, we say that $u$ is a \emph{neighbor} of $v$ if and only if $\varphi(u)\varphi(v) \in E(H)$.
A switching of $u$ and $v$ is taken to be understood as a new placement $\varphi'$ where for each $w \in V(G), w \neq u,v$ $  \varphi'(w)=\varphi(w)$ and $\varphi'(u)=\varphi(v)$ and $\varphi'(v)=\varphi(u)$, i.e., $u$ and $v$ switch places on $H$.

We propose the following greedy strategy. A node switches with a neighbor if, according to the (observed) marginal distribution on the requests involving itself and its neighbor, switching positions improves the objective value.
In this work, we consider two simple optimization rules:

\begin{enumerate}

\item \textbf{M-rule}: Node will switch locations  with its neighbor if the switch will minimize the objective function: the expected path-length between all pairs of nodes. This criterion is exactly the MEPL objective.

\item \textbf{C-rule}: Node will switch location with its neighbors if the switch will minimize the expected path-length between the \emph{center} node and all the other nodes. This objective does not give us a solution for the MEPL problem, but it will be used as an upper bound for it. 

\end{enumerate}

%\michael{Do we need to prove something about the convergence?? or its speed?}
If nodes switch only when this decreases the expected path-length (or some other criterion), then it is clear that this, strictly monotone, potential can not drop indefinitely (or too often) and thus, a (quick) convergence is guaranteed.
A placement $\varphi$ is said to be {\bf \emph{local minimum}} (or local optimum) if and only if no node in $G$ can switch according to the rule they operate (i.e., M-rule or C-rule). When using the C-rule, we can prove the following about the local minimum placement.

\begin{lemma}\label{lem:monoton}
%Any local minimum in C-rule is \emph{center monotone}: If $u$ is a neighbor of $v$ on a shortest path from $v$ to $c^*$ than $\p(u) \ge \p(v)$.
Any local minimum placement $\varphi$ with respect to the C-rule, is \emph{center monotone}: If there is a path of switching directions from $\varphi(u)$ to $\varphi(v)$ that is distance-decreasing with respect to $c^*$, i.e., a path for which every step is a switching direction and goes strictly closer to $c^*$, then $\p(u) \leq \p(v)$.
%{\bfseries Comment Bernhard: We should strengthen this to any path that is distance monotone, i.e., any path for which every step goes strictly closer to $c^*$. This is stronger and in particular also implies this result for any shortest path. We should use this strengthening in the proof of Lemma 8.}
\end{lemma}
\begin{proof}
%Assume by contradiction that $\varphi$ is a local minimum and $u$ is a neighbor  of $v$ on a shortest path from $v$ to $c^*$ but $\p(u) < \p(v)$. Let  $\varphi'$ be the placement  after the switch, we will show $C(\varphi') < C(\varphi)$ contradiction to the local optimality of $\varphi$. Let $C'(\varphi')$ be the expected distance to $c^*(\varphi)$ (old center) in $\varphi'$ (new placement ). Clearly $C'(\varphi')<  C(\varphi)$ since we changed only two routes (of $u$ and $v$) and  $v$ which has higher probability got closer. But now, by definition $C(\varphi') \le C'(\varphi')$ and we done.
Assume for sake of contradiction that $\varphi$ is a local minimum, that $p(u) > p(v)$ and that there is a distance-decreasing path $P$ of switching directions from $\varphi(u)$ and to $\varphi(v)$. By induction there has to be two nodes $u'$ and $v'$ such that $\varphi(u')$ and $\varphi(v')$ are neighbors on the path $P$ with $p(u') > p(v')$ but $v'$ is closer to $c^*$ then $u'$. By assumption it is possible to switch $u'$ and $v'$ and it is easy to see that this is an improvement with regards to the C-rule. This contradicts the assumption that $\varphi$ is a local minimum. 
% $$u$ is a neighbor of $v$ on a distance-decreasing path from $\varphi(v)$ to $c^*$ but $\p(u) < \p(v)$. Let  $\varphi'$ be the placement after the $u\leftrightarrow v$ switch. We will show that $C(\varphi') < C(\varphi)$, which is a  contradiction to the local optimality of $\varphi$. Let $C'(\varphi')$ be the expected distance to $c^*(\varphi)$ (old center) in $\varphi'$ (new placement). Clearly $C'(\varphi')<  C(\varphi)$ since we changed only two routes (of $u$ and $v$) and $v$, which has higher probability, got closer (since the path $\varphi(v)\rightarrow \varphi(u)\rightarrow\cdots\rightarrow c^*$ is distance-decreasing). But now, by definition, $C(\varphi') \le C'(\varphi')$ and thus the contradiction $C(\varphi') < C(\varphi)$ follows.
\end{proof}

Note that according to the C-rule, two neighbors switch locations only if the switch decreases C: the expected distance to the center.
The improvement of the switch can be found  locally, since the center location can be computed locally at each node via the  expected position of its requests (which are identical to all nodes because of the product distribution). Therefore, the C-rule will greedily minimize, for each node, the distance to the expected position of its requests. 

% ====== Add later ??? =================
%Since the decisions have to be made according to some samples from a distribution, the system actually forms a Markov Chain with no sinks and thus never completely converges. Nevertheless if the switches are implemented conservatively enough (i.e., only if a sufficiently large (increasing) threshold is crossed) one can guarantee a quick convergence against a local minimum with high probability and guarantee (if one increases the threshold slightly over time) even a indefinitely stable state with probability one.

Throughout the rest of this paper we assume that the system converges against a local minimum and analyze the performance of such a solution in this stable state. On the other hand, we do NOT assume anything about the starting position OR the specific order of the dynamics (node switches). Thus, in many cases, an initially random starting position converges (e.g., using random improving switches) to a (near) optimal solution; we make no such assumptions and assume a worst case sequence of improvements and a worst-case initialization.

\subsection{The Line - Linear placement}
First, we study a greedy local switch strategy on a $1$-dimensional grid - the line.
We assume that the C-rule switching strategy is sequentially applied (in arbitrary order) on an arbitrary initial state and continuously adjust the network by switching neighbors. The strategy will converge against a local optimum from which no switch of two neighboring nodes improves the objective value in expectation.
We are interested in quantifying how far such a locally optimal solution can be from the global optimum. The following theorem gives an answer for this question.

\begin{theorem}\label{thm:approx-line}
Let $H$ be the line and $\Pc$ a symmetric product distribution, then any locally optimal solution achieved by the C-rule (or M-rule) is at most a factor of four larger than the global optimum of MEPL.
\end{theorem}

We prove this theorem for the C-rule but this could be done similarly to the M-rule. Assume $H$ and $\Pc$ as in the theorem.
We first give an upper bound on the expected path length achieved by the C-rule in terms of the expected rank of the distribution (Definition~\ref{dfn:rank}).

\begin{lemma}\label{lemma:upper_bound_line}
For any locally optimal solution $\varphi$ achieved by the C-rule on $C(\varphi) \le \E[R]$ and $\epl(\varphi)$ is at most $2\E[R]$.
\end{lemma}
\begin{proof}
Let $\d(\varphi(v),c^*)$ be the distance of $\varphi(v)$ from $c^*$. We want to bound it in terms of $r_v$, the rank of $v$.
From Lemma~\ref{lem:monoton} we get that all nodes between $\varphi(v)$ and $c^*$ on the line have higher probability than $v$ and thus $\d(\varphi(v),c^*) \le r_v$.
\begin{align*}
C(\varphi)&=\sum_{v\in V(G)}\p(v)\d(\varphi(v),c^*)\le \sum_{v \in V(G)} \p(v) \r(v)=\E[R]
\end{align*}
From Lemma~\ref{lem:epl-c}, $\epl(\varphi) \le 2C(\varphi) \le 2\E[R]$.
\end{proof}
We now prove a lower bound for MEPL on the line and any symmetric product distribution of requests.
\begin{lemma}\label{lem:1-D-lower}
$\mepl \ge C_{\min} \ge \frac{1}{2}\E[R]$.
\end{lemma}
\begin{proof}
Let $\varphi^*$ be the placement  such that $C_{\min} = C(\varphi^*)$. Note that by definition, $\varphi^*$ minimizes the expected path length to the center. Given the center $c^*(\varphi^*)$ and an arbitrary node $v$ with a distance $\d(\varphi^*(v),c^*)$ from $c^*$, we want to find an upper bound on the rank of $v$ by bounding how many nodes can have a larger activity level than $v$. Clearly, all such nodes will be at most at the distance $\d(\varphi^*(v),c^*)$ from the center, since otherwise, $\varphi^*$ will not be optimal.
%and $v$ will be switch to a worse location. 
Since in a line there at most two nodes at distance $i$ from the center $\d(\varphi^*(v),c^*) \ge \r(v)/2$ we obtain as desired: $C_{\min}=\sum_{v\in V}\p(v)\d(\varphi^*(v),c^*)\ge\sum_{v}\p(v) \r(v)/2=\E[R/2]$.
%\begin{align*}
%C_{\min}&=\sum_{v\in V}\p(v)\d(v,c^*)\ge\sum_{v}\p(v) \r(v)/2=\E[R/2]
%\end{align*}
\end{proof}

To conclude the proof of Theorem~\ref{thm:approx-line}, we combine Lemma~\ref{lem:epl-c}, Lemma~\ref{lemma:upper_bound_line} and Lemma~\ref{lem:1-D-lower} to get that for a local minimum $\varphi$:
\begin{align}
2\E[R] \ge \epl(\varphi) \ge \mepl \ge \frac{1}{2}\E[R]
\end{align}
Thus, the ratio between the worst case local solution and the optimal solution is at most $4$.  
%\end{proof}

\subsection{The d-Dimensional Grid}

In this section, we extend the ideas from the line to grid networks. Our results apply readily to grids of arbitrary dimension but, for sake of simplicity, we stick to two dimensions here.
%al case in which each internal node is connected to the four neighbors of $\ell_1$-distance one. 
%One can equally well also look at the network with diagonal connections to the $\ell_\infty$-distance 1 neighbors.
We first show that using the same greedy approach as in the line, namely switching neighboring nodes using the M-rule, leads to a drastically worse ratio between local and global minima.

\begin{lemma}
On the $d$-dimensional grid, there is a local minimum with regards to the M-rule and the C-rule that is a factor of $\Omega(n^{1/d - 1/d^2})$ worse than the global minimum.
\end{lemma}
\begin{proof}
We take $n^{1/d}$ nodes with uniform probability $p = n^{-1/d}$ and set all other nodes to probability zero. We arrange the active nodes on a line along one dimension as an initial placement. We now note that this initial placement is a local minimum since all switches are non-improving with regards to the M-rule or C-rule and also lead to a higher expected path length. This is true because any switch between nodes in the line does not change the expected path length while switching a node from the line with an inactive node only increases the path length for the active node by one while the reduced path length for the inactive node has no effect on the expected path length. The same is true for the expected distance to the center. The expected path length of this linear arrangement is $n^{1/d}/2$. To complete the proof we will now argue that there are arrangements with EPL of order $O(n^{1/d^2})$. One such (close to optimal) solution can be achieved by arranging all active nodes within a ball around the center. Since the number of nodes in a ball in $d$ dimensions grows as $r^d$ with the radius $d$ it is clear that a radius of $r = O(n^{1/d^2})$ suffices to contain enough spots to place the $n^{1/d}$ active nodes in. The maximum (and also expected) path length within the ball is also within a constant factor of the radius, i.e., also $O(n^{1/d^2})$. The ratio between the local minimum and the alternative ball arrangement is thus $\Omega(n^{1/d - 1/d^2})$ proving that the ratio between the worst local minimum compared to the global minimum is at least this large. 
\end{proof}

%\begin{figure}[htbp]
%\begin{center}
%	\includegraphics[width=0.45\textwidth]{pics/optimum}\label{fig:opt-ball}
%\caption{Figure~\ref{fig:opt-ball}: The optimal solution for minimizing the expected path length between the red nodes in the grid. Note that while the path length is measured with respect to the $\ell_1$ distance the solution is closer to an $\ell_2$ ball in the plane to keep the nodes closer together.} 
%\end{center}
%\end{figure}

Note that the last lemma implies a $\Omega(n^{1/4})$ worst-case ratio for the $2$-dimensional grid. Surprisingly, we can avoid this, locally stable but highly suboptimal solution, by allowing only slightly longer switches. The rule we propose is that a node can also switch with any of the neighbors in $\ell_1$-distance three that differs two in one axis and one in the other (similar to the chess knight moves) and the switching is according to the C-rule. In this case we can prove the following.

\begin{theorem}\label{thm:approx-grid}
Let $H$ be the $2$-dimensional grid and $\Pc$ a symmetric product distribution, then any locally optimal solution achieved by the C-rule (and allowing "chess knight move" switches) is at most a factor of $4.62$ larger than the global optimum of MEPL.
\end{theorem}

The proof of this theorem is similar in spirit to the $1$-dimensional grid, where we provide bounds on $C$ for the optimal and locally optimal placements.
First, we show that we get a fat set from this strategy and also prove a stronger rank-property of any locally optimal solution; namely, nodes that are far away from the center have to have a (quadratically) high rank. 

%\begin{figure}
%\begin{center}
%	\includegraphics[width=0.45\textwidth]{pics/local-minimum-optimize-distancetocenter}\label{fig:local-min-center}
%\caption{Figure~\ref{fig:local-min-center}: A local minimum when optimizing the distance to the expected position of a routing request (black node). The red node mark the positions that dominate the node in position green, i.e., they have a path of switch directions that go all strictly closer to the black node. These nodes have therefore a smaller probability than the green node making the rank of the green node at least as large as the number of red nodes.}
%\end{center}
%\end{figure}

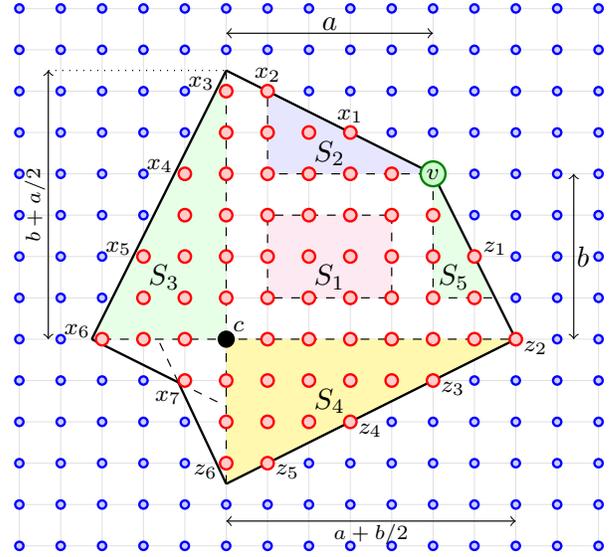
\begin{figure}
\begin{center}
%\scalebox{0.8}
\scalebox{1.1}{\begin{tikzpicture}
[inner sep=0.3mm, place/.style={circle,draw=green!50!black,fill=green!20,thick,minimum size=3mm},place1/.style={circle,draw=magenta,fill=green!20,thick,minimum size=3mm},place2/.style={circle,draw=red,fill=red!20,thick,minimum size=1.5mm},place3/.style={circle,draw=blue,fill=blue!20,thick,minimum size=1mm},place4/.style={circle,draw=black,fill=black,thick,minimum size=1.8mm},scale=1]

\draw[very thin,color=gray,scale=0.5,opacity=0.25] (2,1) grid (16,14);

\foreach \y in {0.5,1,...,7}{
\foreach \x in {1,1.5,...,8}{
\node at (\x,\y) [place3] {\scriptsize{$$}};
}}

\filldraw [draw=black,dashed,fill=blue!10] (6,5)--(4,5+1)--(4,5)--(6,5);
\filldraw [draw=black,dashed,fill=green!10] (6,5)--(6,3.5)--(6.75,3.5);
\filldraw [draw=black,dashed,,fill=yellow!40] (3.5,3-1.75)--(3.5,3)--(7,3);
\filldraw [draw=black,dashed,,fill=green!10] (3.5,5+1.25)--(3.5,3)--(3.5-1.625,3);
%\filldraw [draw=white,fill=magenta!10] (3.5,3-1.75)--(2.9217,2.47664)--(3.5-1.625,3)--(3.5,3);
\filldraw [draw=black,dashed,fill=magenta!10] (5.5,4.5)--(5.5,3.5)--(4,3.5)--(4,4.5)--(5.5,4.5);

\node at (4.75,3.75) {$S_1$};
\node at (4.75,5.25) {$S_2$};
\node at (2.75,3.75) {$S_3$};
\node at (4.75,2.25) {$S_4$};
\node at (6.25,3.75) {$S_5$};

\node at (5,5.7) {\footnotesize{$x_1$}};
\node at (4,6.2) {\footnotesize{$x_2$}};
\node at (3.2,6.1) {\footnotesize{$x_3$}};
\node at (2.7,5.1) {\footnotesize{$x_4$}};
\node at (2.2,4.1) {\footnotesize{$x_5$}};
\node at (1.7,3.1) {\footnotesize{$x_6$}};
\node at (2.8,2.3) {\footnotesize{$x_7$}};
\node at (3.25,1.4) {\footnotesize{$z_6$}};
\node at (4.24,1.4) {\footnotesize{$z_5$}};
\node at (5.24,1.9) {\footnotesize{$z_4$}};
\node at (6.24,2.4) {\footnotesize{$z_3$}};
\node at (7.24,2.9) {\footnotesize{$z_2$}};
\node at (6.75,4.1) {\footnotesize{$z_1$}};

\draw [<->] (6,6.7)-- node[above] {$a$}(3.5,6.7);
\draw [<->] (7.7,5)-- node[right] {$b$}(7.7,3);
\draw [<->] (3.5,0.8)-- node[below] {\scriptsize{$a+b/2$}}(7,0.8);
\draw [<->] (1.35,3)-- node[above,rotate=90] {\scriptsize{$b+a/2$}}(1.35,6.25);
\draw [dotted] (1.35,6.25)-- (3.5,6.25);

\draw[thick] (6,5)--(3.5,5+1.25);
\draw[thick] (6,5)--(7,3);
\draw[thick] (7,3)--(3.5,3-1.75);
\draw[thick] (3.5,3-1.75)--(2.9217,2.47664);
\draw[dashed] (2.9217,2.47664)--(3.5-0.825,3);
\draw[thick] (3.5,5+1.25)--(3.5-1.625,3);
\draw[thick] (3.5-1.625,3)--(2.9217,2.47664);
\draw[dashed] (2.9217,2.47664)--(3.5,3-0.8125);
%\draw[dashed] (3.5,5+1.25)--(3.5,3-1.75);
%\draw[dashed] (6,5)--(3.5,5);
%\draw[dashed] (6,5)--(6,3);
%\draw[dashed] (3.5-1.625,3) -- (7,3);

%\node at (3.25,2.75) {\footnotesize{$c ~\bullet$}};

\node at (6,5) (v)[place] {\footnotesize{$v$}};

\foreach \y in {5,4.5,4,3.5,3}{
\foreach \x in {5.5,5,4.5,4,3.5}{
\node at (\x,\y) [place2] {\scriptsize{$$}};
}}

\node at (3.5,3) [place4] {\footnotesize{$$}};
\node at (3.65,3.15)  {\footnotesize{$c$}};

\foreach \y in {4.5,4,3.5,3}{
\node at (6,\y) [place2] {\scriptsize{$$}};
}

\foreach \x in {5,4.5,4,3.5}{
\node at (\x,5.5) [place2] {\scriptsize{$$}};
}

\foreach \y in {2.5,2,1.5}{
\node at (3.5,\y) [place2] {\scriptsize{$$}};
}

\foreach \x in {3,4,4.5,5,5.5,6}{
\node at (\x,2.5) [place2] {\scriptsize{$$}};
}

\foreach \x in {4,4.5,5}{
\node at (\x,2) [place2] {\scriptsize{$$}};
}

\foreach \x in {4}{
\node at (\x,1.5) [place2] {\scriptsize{$$}};
}

\foreach \y in {5,4.5,4,3.5,3}{
\node at (3,\y) [place2] {\scriptsize{$$}};
}

\foreach \y in {4,3.5,3}{
\node at (2.5,\y) [place2] {\scriptsize{$$}};
}

\node at (2,3) [place2] {\scriptsize{$$}};
\node at (6.5,3) [place2] {\scriptsize{$$}};
\node at (7,3) [place2] {\scriptsize{$$}};
\node at (6.5,3.5) [place2] {\scriptsize{$$}};
\node at (6.5,4) [place2] {\scriptsize{$$}};

\node at (4,6) [place2] {\scriptsize{$$}};
\node at (3.5,6) [place2] {\scriptsize{$$}};

%\node at (5,5.5)[place2]  {\footnotesize{$x_1$}};
%\node at (4,6) [place2]{\scriptsize{$x_2$}};
%\node at (3.5,6)[place2] {\tiny{$x_3$}};

\end{tikzpicture}}
\caption{\label{fig:local-min-center}
%Figure~\ref{fig:local-min-center}: 
A local minimum when optimizing the distance to the expected position of a routing request (black node). The red nodes mark the positions that dominate the green node, that is, that have a path consisting of switch directions each going strictly closer to the black node. In a local minimum these nodes have to have a larger probability than the green node making the rank of the green node at least as large as the number of red nodes.}
\vspace{-0.4cm}
\end{center}
\end{figure}

\begin{lemma}\label{lemma:c-rule-upper}
%Any node of distance $k$ to the center has rank at least $\frac{9}{16}k^2$.
For any local minimum of the C-rule $C(\varphi) \le  \frac{4}{\sqrt{6}} \E[\sqrt{R}]$, where $R$ is the rank of a given distribution (see Definition~\ref{dfn:rank}). 
%Formally: $C=\sum_{v\in V}\p(v)\|c^*-{\psi}_v\|\le \E[\frac{4}{3}\sqrt{R}]$.
%\\{\bfseries Comment Bernhard: R is not defined here / it is not clear that the rank is meant.}
\end{lemma}
\begin{proof}
%{\bfseries Comment Bernhard: The first two paragraphs need to be rewritten for language. Also we should change the proof to say that all nodes in the polytope are on strictly distance monotone paths instead of shortest paths. (Shortest paths happens to be true here but only because of a weird artefact of 2,1 switches. It would not be true for 3,1 switches for example). See also Lemma 4
%We should strengthen this to any path that is distance monotone, i.e., any path for which every step goes strictly closer to $c^*$. This is stronger and in particular also implies this result for any shortest path. We should use this strengthening in the proof of Lemma 8.}
%
%For a node $v$ in general relative position to the center $c^*$ we analyze the area
%%%%%%(which is almost equals to the number of nodes inside it) 
%that can be reached via a path of improvement switch directions that each go strictly closer to the center. The edges of the polygon in Figure~\ref{fig:local-min-center} bound the area where located nodes that have larger activity levels than $v$. The last is true due to  Lemma~\ref{lem:monoton} that states that for any local minimum in C-rule, if $u$ is a neighbor of $v$ on a shortest path from $v$ to $c^*$, then $\p(u) \ge \p(v)$.
%
We consider a local minimum of the C-rule achieved by $\varphi$. Without loss of generality we assume that $\varphi$ is the identity function. Consider a node $v$ in general relative position to the center $c^*$ (see Figure~\ref{fig:local-min-center}). We want to estimate how many nodes have a larger probability (higher rank) than the node $v$. To achieve this estimation, we analyze the area of the largest polygon, such that every node inside and on the edges of the polygon belongs to some distance-decreasing path from $v$ to $c^*$. According to Lemma~\ref{lem:monoton} we then get the guarantee that for any local minimum with respect to the C-rule any node $u$ in the polygon (with a distance decreasing path from $v$) has $\p(v) \leq \p(u)$. Figure~\ref{fig:local-min-center} provides an example for this: The node $x_1$ is closer to the center than $v$ and switching between $x_1$ and $v$ is possible. This implies that $p(v) \leq p(x_1)$. Furthermore, if we look at the distance-decreasing paths $v,x_1,x_2,\ldots,x_7$ and $v,z_1,z_2,\ldots,x_7$ we obtain from Lemma~\ref{lem:monoton} that $\p(v) \leq \p(x_1) \leq p(x_2) \leq \ldots \leq p(x_7)$ and that $\p(v) \leq \p(z_1) \leq \ldots \leq p(x_7)$. These paths can furthermore be extended to any node in the polygon. All these nodes have such higher activity levels then $p(v)$. 

%For the nodes that are close to the boundary of the polygon holds:
%\begin{align*}
%\d(v,c^*)&>\d(x_1,c^*)>\d(x_2,c^*)>\ldots >\d(x_7,c^*)\\
%\text{and}\\
%\d(v,c^*)&>\d(z_1,c^*)>\d(z_2,c^*)>\ldots >\d(x_7,c^*)
%\end{align*}

%If $\p(x_1)<\p(v)$ then the switch $v \leftrightarrow x_1$ would improve the objective function of C-rule -- expected distance to the center ($C=\sum_{v\in V}\p(v)\d(v,c^*)$). The last is true since if the $c^*$ was not changed after the $v \leftrightarrow x_1$ switch, then $\sum_{v\in V}\p(v)\d(v,c^*)$ definitely became smaller since the larger probability is now assigned to the shorter distance. But if the center $c^*$ was moved it is because the expected distance to the center -- $C$ can be smaller than without moving the center ($C$ is defined as the expected distance to the optimal center location). Thus, whenever we perform a switch that puts a node with a larger probability closer to the current center, the expected distance to the center (that possibly changed location due to the switch) get smaller.

%The last observation implies that if the network is in a local minimum of the C-rule, all nodes that can be physically switched with $v$ and closer than $v$ to the center have larger probability than $v$. Now, using an induction argument we can say that the same about these nodes, e.g., since 

To get a bound on the rank of $v$, i.e., in the number of nodes that have a larger activity level, we count the number of nodes $S_{total}$ that are inside the polygon. This number generally involves many floors and ceilings. We avoid these by first calculating the number of nodes $A(x)$ bounded by a right triangle shape that starts at a point and whose (axis parallel) legs have length $x$ and $x/2$. We denote this number by $A(x)$ and in the following give a formula and an estimate for it that holds for any positive real $x$:
\begin{align*}
A(x):=\sum_{i=0}^{\left\lfloor x/2\right\rfloor}(\left\lfloor x\right\rfloor + 1 -2i)\ge
\begin{cases} \frac{x^2}{4}+1 & \text{if $x\ge1$,}
\\
1 & \text{if $x<1$.}
\end{cases}
\end{align*}

%For any locally optimal C-rule solution, and given an arbitrary node $v$, we can determine the amount of nodes that have larger probability than $v$. These nodes are surrounded by lines of improvement directions of the node $v$. Since the node $v$ did not switch with the nodes along (and inside) these directions, it is because its probability is smaller and thus the switch would increase the objective function.
%
%Therefore, we get that the rank of the node $v$ is at least the number of nodes located in this area.

Now we are ready to calculate $S_{total}$. We denote the number of nodes in the middle rectangle as $S_1$, the number of nodes in the upper triangle as $S_2$ and so on according to Figure~\ref{fig:local-min-center}.
\begin{align*}
S_1&=(a-1)(b-1) &S_2=A(a-1) && S_3=A(b+a/2)\\
S_4&=A(a+b/2)   &S_5=A(b-1) 
\end{align*}
So, we obtain that $S_{total}=S_1+S_2+S_3+S_4+S_5-3$, where $-3$ is needed since the node $v$ should not be calculated (but was counted twice) and $c^*$ should be calculated once (but was counted twice).
By adding up the expressions we obtain for $a,b\ge2$:
\begin{align*}
S_{total}&=S_1+S_2+S_3+S_4+S_5-3\\
&\ge(a-1)(b-1)+ \frac{(a-1)^2}{4} + \\
&+\frac{(b+\tfrac{a}{2})^2}{4}+ \frac{(a+\tfrac{b}{2})^2}{4}+\\
&+ \frac{(b-1)^2}{4}+1
\ge\tfrac{6}{16}(a+b)^2
\end{align*}
%The box with the green node $v$ and the black node (near $c^*$) as diagonal contains $(a+1)\cdot(b+1)$ nodes. We also have six right triangles with the legs ratio $2:1$. This ratio is true due to the improvement steps for the node $v$, e.g., get closer to the $c^*$ by 2 in the $x$ coordinate and get farther by 1 in the $y$ coordinate -- is an improvement.
%The number of nodes inside the right triangle with the legs ratio $2:1$, and the small leg equals $m$, is:
%\begin{align*}
%A(m):&=\sum_{i=1}^{\left\lfloor m/2\right\rfloor}(m-2i)
%\ge\tfrac{1}{4}(m^2-m)
%\end{align*}
%The total number of nodes (which is also the rank of $v$) is thus at least (ignoring the two smallest of the six triangles):
%
%\begin{align*}
%\r(v)\ge(a+1)(b+1)&+ A(a)+ A(b) \\
%&+ A(b+\left\lfloor a/2\right\rfloor)\\
%&+A(a+\left\lfloor b/2\right\rfloor)\\
%&\ge \tfrac{9}{16}(a+b)^2=\tfrac{9}{16}(k_v)^2
%\end{align*}
Easy to verify that the inequality $S_{total}\ge\tfrac{6}{16}(a+b)^2$ holds also in the case where $a$, or $b$, or both equal to $1$.
Since the rank of $v$ is at least $S_{total}$, we get $\r(v)\ge\tfrac{6}{16}(\d(v,c^*))^2$, and thus: $\d(v,c^*)\le\tfrac{4}{\sqrt{6}}\sqrt{\r(v)}$. Now,
\begin{align*}
C&=\sum_{v\in V}\p(v)\d(v,c^*)\le\sum_{v\in V}\p(v) \tfrac{4}{\sqrt{6}}\sqrt{\r(v)}
=\tfrac{4}{\sqrt{6}}\E[\sqrt{R}].
\end{align*}
\end{proof}

%\begin{corollary}
%Any local minima of the move to the expected center strategy has a cost of at most $2 E[\frac{4}{3} \sqrt{R}]$ where $R$ is the random variable denoting the rank of a sampled node.
%\end{corollary}
%\begin{proof}
%Take the bound on the rank from Lemma~\ref{lem:rankbound-localgrid} and use Proposition~\ref{prop:approximation} to translate the expected $\frac{4}{3} \sqrt{R}$ distance to the center with a slack of two into an upper bound for the expected path length. 
%\end{proof}

Next we prove a lower bound for the cost of the optimum placement obtaining a similar expression in terms of the (expected) rank.

\begin{lemma}\label{lem:c-rule-lower}
%Any solution has expected path length at least $E[\sqrt{\frac{R}{2}}]$.
$\mepl \ge C_{\min} \ge \frac{1}{\sqrt{2}}\E[\sqrt{R}]$.
%Even if the nodes arrangement is $\Psi_{gmin}$, where $\Psi_{gmin}=\arg \min_{\Psi}C(\Psi)$, the value of $C(\varphi_{gmin})$ is at least $\E[\sqrt{R/2}]$.
%Any local minimum achieved by the C-rule strategy has a cost of at least $ E[\sqrt{R/2}]$ where $R$ is the rank of a given distribution. Formally: $C=\sum_{i}p_i\|c^*-{pos}_i\|\ge E[\sqrt{R/2}]$.
\end{lemma}
\begin{proof}
Let $\varphi^*$ be the placement  such that $C_{\min} = C(\varphi^*)$. Note that by definition $\varphi^*$ minimize the expected path to the center. Given the center $c^*(\varphi^*)$ and an arbitrary node $v$ with a distance $\d(v,c^*)$ from $c^*$, we again find an upper bound on the rank of $v$ by showing how many nodes can have a larger activity level than $v$. Again, all such nodes will be at most at the distance $\d(v,c^*)$ from the center, since otherwise the solution will not be a global optimum. There are exactly $4$ nodes at distance $1$ from $c^*$, $8$ nodes at the distance $2$ and in general $4i$ nodes at distance i. This lead to the to the following bound:
\begin{align*}
\r(v)&\le\sum_{i=1}^{\d(v,c^*)} 4i=2(\d(v,c^*))^2
\end{align*}
So, we obtain $\d(v,c^*)\ge\sqrt{\r(v)/2}$, and thus:\\
\begin{align*}
C&=\sum_{v\in V}\p(v)\d(v,c^*)=\sum_{v\in V}\p(v) \d(v,c^*)\\
&\ge\sum_{v}\p(v) \sqrt{\r(v)/2}=\E[\sqrt{R/2}].
\end{align*}
\end{proof}

Now we are ready to prove the result of Theorem~\ref{thm:approx-grid}. From Lemma~\ref{lem:epl-c} and~\ref{lem:c-rule-lower} we get that the minimum expected path length is at least $\E[\sqrt{R/2}]$. From Lemma~\ref{lemma:c-rule-upper} we obtain that $C(\varphi) \le \E[\frac{4}{\sqrt{6}} \sqrt{R}]$, and thus by Lemma~\ref{lem:epl-c} the expected path length of a local minimum can not be larger than $2\E[\frac{4}{\sqrt{6}} \sqrt{R}]$. Therefore the ratio between the optimal solution and any local minimum with regards to the C-rule is at most $2 \frac{4\sqrt{2}}{\sqrt{6}} \approx 4.62$.

% ---------- Chen: if we have space ------------------
%All proofs above can be extended to the $d$-dimensional grid. The lower bound guarantees in this context that any solution on the $d$-dimensional grid has a cost of at least $\Omega(E[R^{1/d}])$, where $R$ is the rank of a sampled node. Similarly for any constant $d$ we get a $O(E[R^{1/d}])$ upper bound from the fact that a fat body is dominated by any node. The constant factor in the upper bound decreases like $2^{-d}$ since using the $(2,1)$-improvement directions along any dimension costs a factor of $2$. By using longer improvement directions of length $k$ instead of $3$ the factor of two can be brought down to $1 - \frac{1}{k-2}$. Thus, e.g., using length $d$ improvement directions on the $d$-dimensional grid results in a constant factor approximation for any dimension $d$. 

\section{Clustered Requests} \label{sec:cluster}

We have demonstrated that self-adjusting networks and their local switching rules work well on grid networks with symmetric product distributions. We now briefly discuss some interesting preliminary results on a more general type of request distributions: clustered requests. For this we consider situations where processes can be clustered into groups such that communication takes place only or predominantly between processes belonging to the same group. This locality is inspired by practice and we believe that such a structure in the requests is quite common.

Ideally, a self-adjusting network ``detects'' such clusters and arranges processes such that groups will reside in separate parts of the network infrastructure. Such an arrangement facilitates short routes since requests between group members get routed quickly without leaving the group. Such a well clustered placement of nodes can be a drastic improvement of the expected path length compared to a non-optimized placement. In particular, any placement that is oblivious to the clustering, e.g., a random placement, will have a bad performance -- leaving plenty of room for improvement. We believe that the simple switching strategies presented in this paper perform very well in many such settings.

Our investigations and simulations on $d$-dimensional torus topologies (see Figure~\ref{fig:clusters_sim}) have led to several interesting preliminary results in this direction: On the negative side we were able to construct local minima in many topologies that have a poor performance compared to the global optimum. This shows that it is not possible to give the same type of strong approximation guarantee independent of the initialization. In simulations we observed that for $d=1$, i.e., on a ring topology, even random initializations lead to bad local minima. The reason for this is that connectivities in the ring are too restricted to allow the resolution of distributed clusters without disturbing other, already fixed, clusters. Fortunately, this changes drastically for higher dimensions. For  $d>1$ local minima become extremely unstable and sensitive to small perturbations. They can only occur if one starts in a carefully constructed worst-case initialization. This observation is supported by our simulations which produce very promising results. Figure~\ref{fig:clusters_sim} shows the outcome of some of these simulations in a $2$-dimensional torus network. We plan to investigate the performance of our self-adjusting networks on this and other topologies in future work.

\subsection{Simulation Results for 2-D Torus}\label{app_simu_torus}

In this section we show simulation results for a $2$-dimensional torus topology. We assume a uniform clustered distribution of requests with uniform clusters' activities. 

Figure~\ref{fig:clusters_sim} shows an example for the improvement in a $2$-dimensional torus.
Overall, the greedy switching algorithm successfully detects clusters and groups them together. However, there are still some clusters that are not connected. In Figure~\ref{fig:clusters_sim} (a) we can see the initial random placement of the clustered nodes. Nodes with a black bold frame are the centers of their clusters (as was defined earlier, a center is a node that has the minimal expected distance to all other nodes in the cluster). In Figure~\ref{fig:clusters_sim} (b) we see the placement of the nodes achieved by the greedy M-rule switching strategy. Although it looks that the nodes are highly grouped, we can see that the shapes of the clusters are not optimal (an optimal placement should look like a circle around the center of the cluster). Some clusters are stretched (e.g., brown cluster) and some are even not connected (e.g, orange cluster). 

When every node on the torus belongs to an active cluster, we can frequently run into situation in which two nodes will not switch even if it is improvement for one of the clusters. This suboptimal local solution can be improved if we allow some nodes on a torus to be inactive. In the following two figures we see the results of such simulation where 50\% on the nodes are inactive. In Figure~\ref{fig:clusters_sim} (c) we can see the initial random placement of the clustered nodes, (the inactive nodes are white colored). In Figure~\ref{fig:clusters_sim} (d) we see the placement of a local minimum of EPL achieved by the same greedy switching strategy. But now we can observe much nicer concentration of the nodes around their centers. These figures lead to many interesting future research questions about the topic. 
At the Figure~\ref{fig:clusters_sim} (c) there is a QR link to the animated version of the figures.

\begin{figure}
\begin{center}
\includegraphics[width=3.4in]{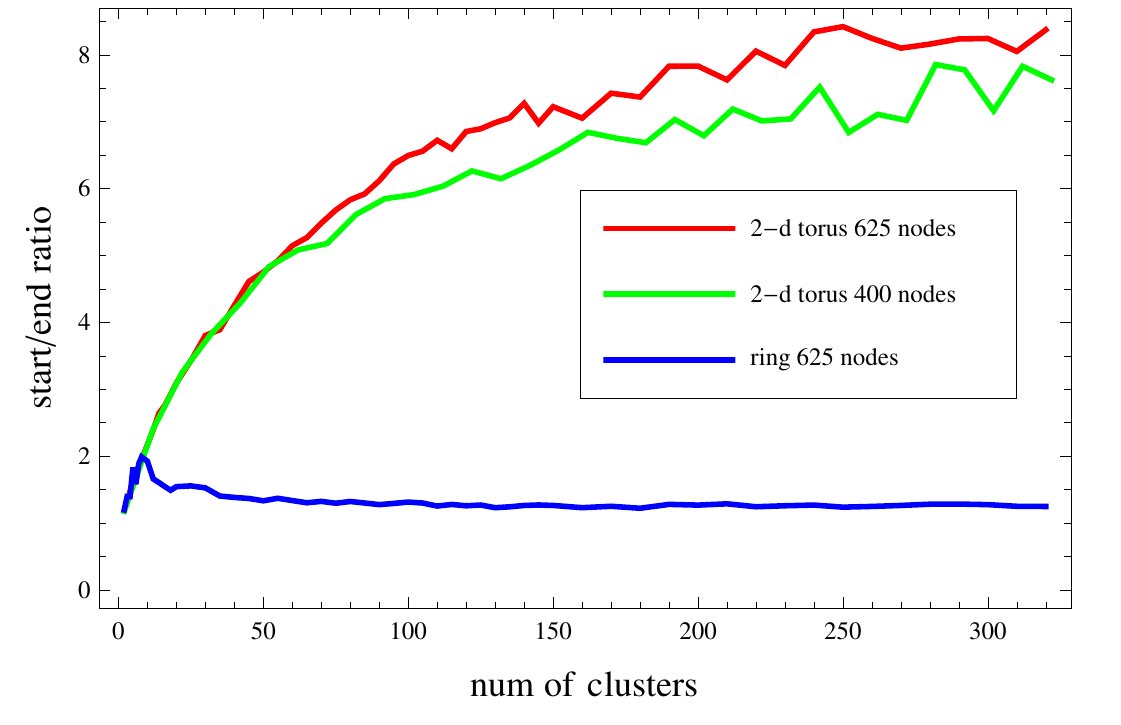}\end{center}
\caption{\label{fig:start_end_ratio_ring_torus_1}Greedy switching on Ring and $2$-dimensional torus. Ratio between the initial and the final EPL, as a function of number of clusters.}
\end{figure}

\begin{figure*} 
\centering
\begin{tabular}{ccc}
\begin{minipage}{3in}
\includegraphics[width=2.8in]{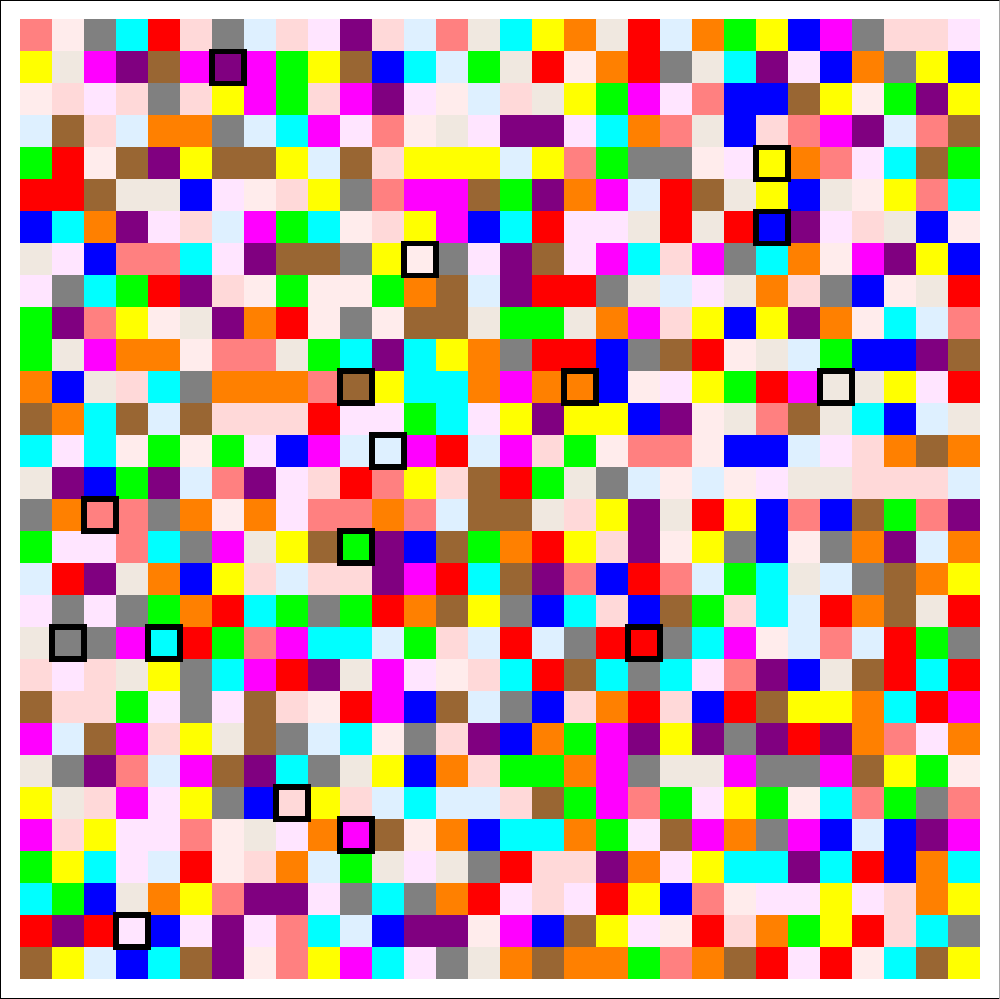}
\end{minipage}&
\begin{minipage}{3in}
\includegraphics[width=2.8in]{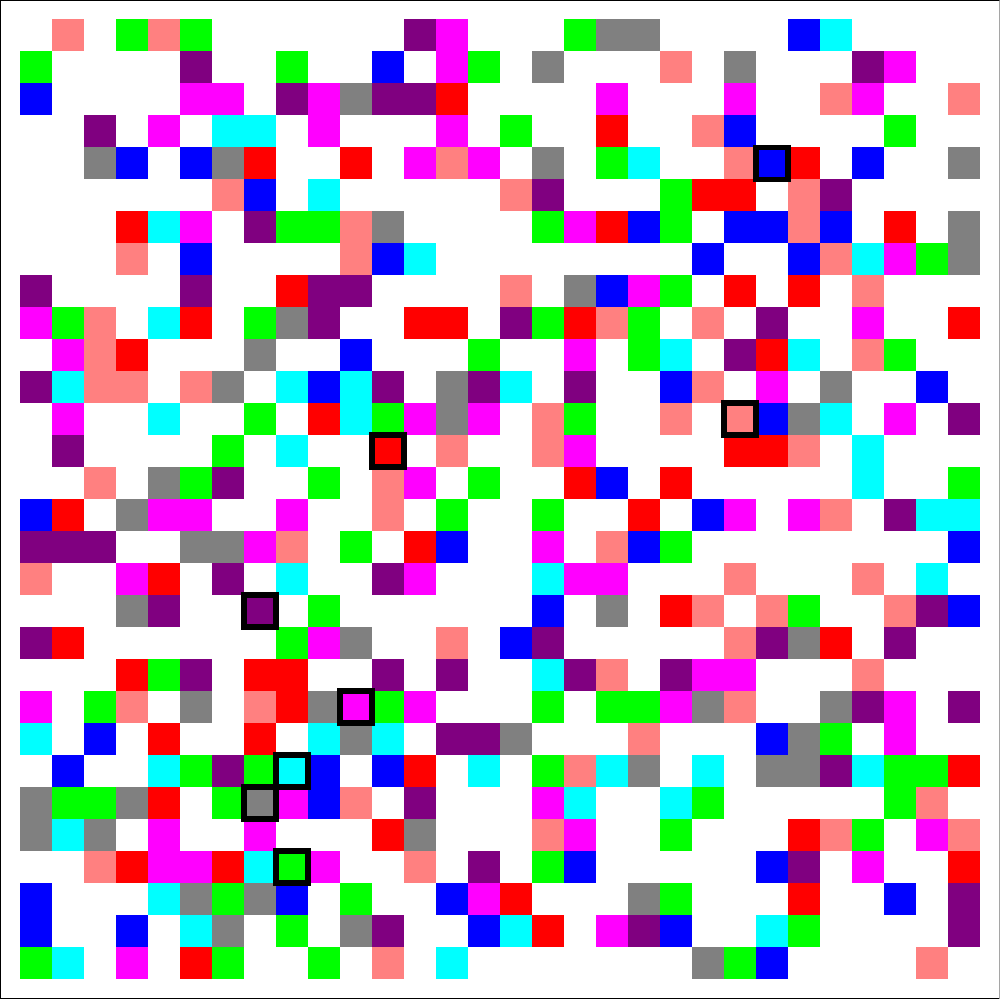}
\end{minipage}&
\multirow{13}{*}{\hspace{-4mm}\includegraphics[width=1in]{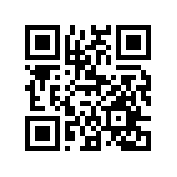}}\\
(a) -- Initial random placement. $\left|V_0\right|=0$. & (c) -- Initial random placement.  $\left|V_0\right|=n/2$. & \hspace{-4mm}(e)\vspace{2mm}\\
\begin{minipage}{3in}
\includegraphics[width=2.8in]{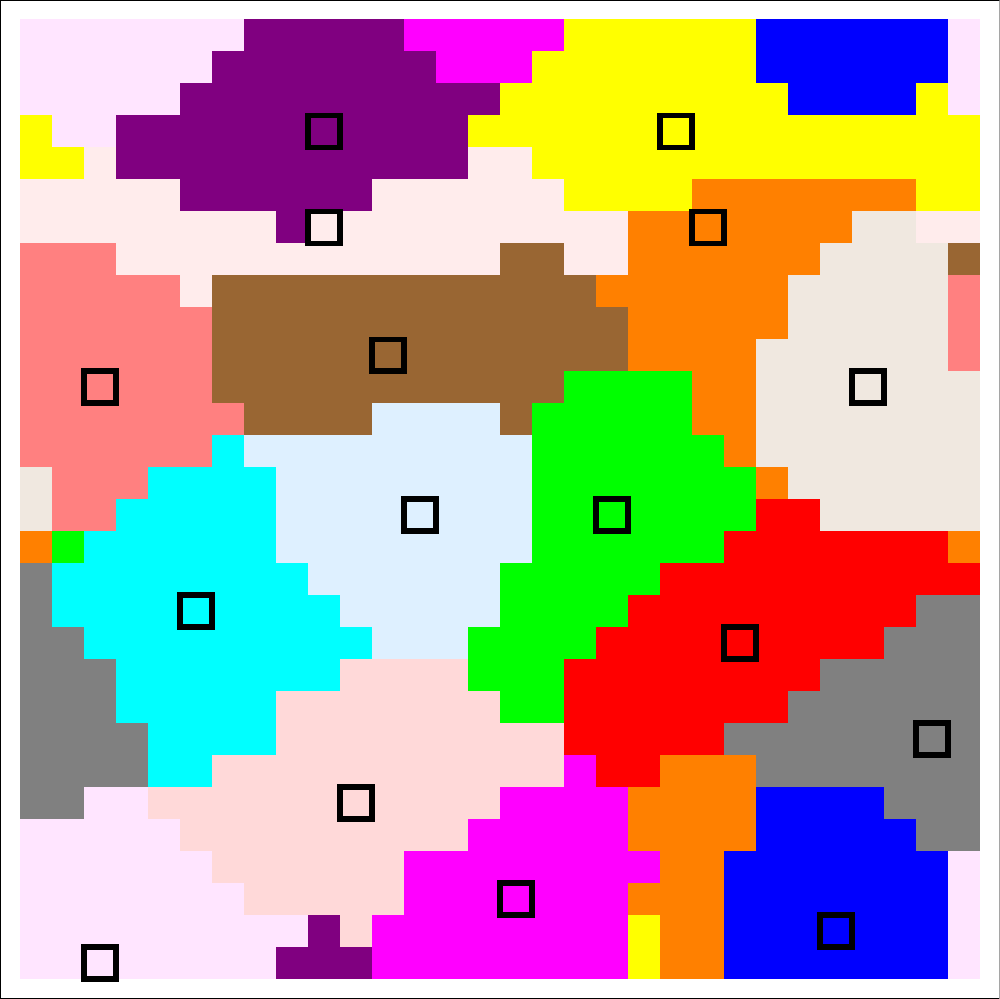}
\end{minipage}&
\begin{minipage}{3in}
\includegraphics[width=2.8in]{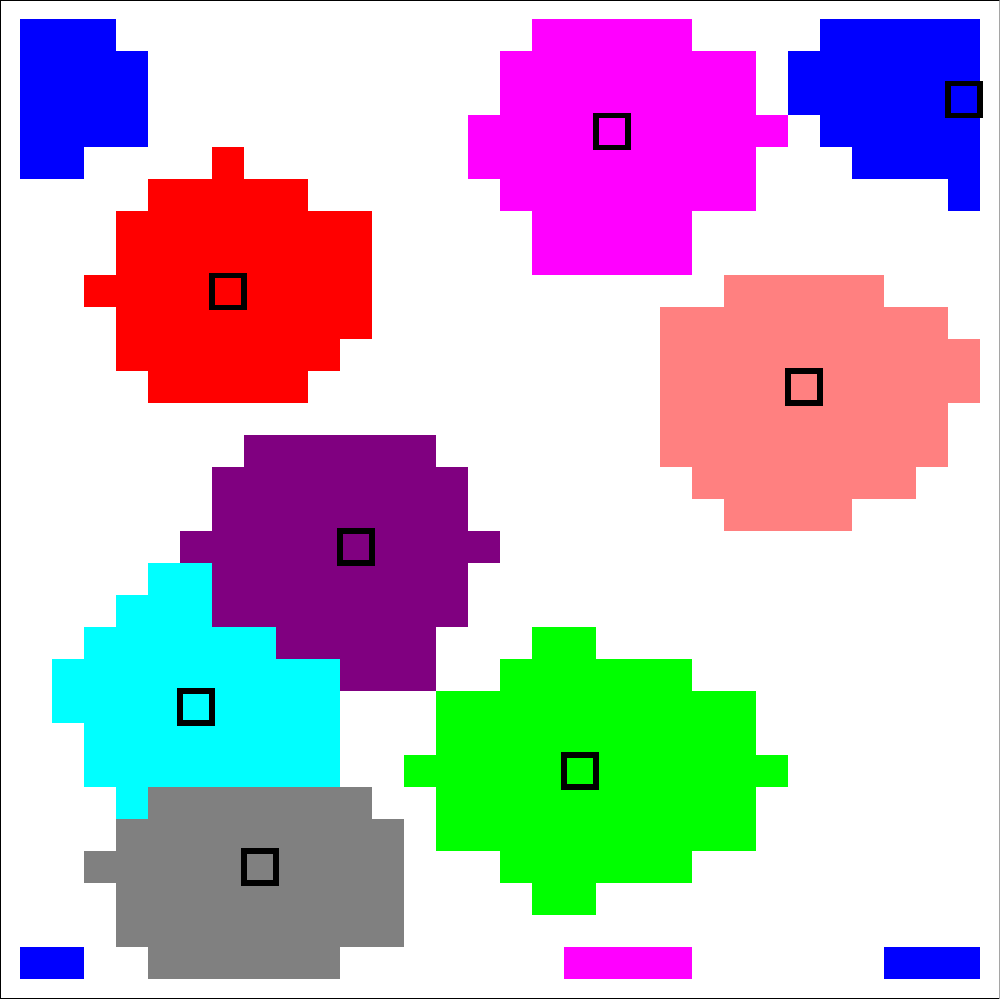}
\end{minipage}\\
(b) -- Final placement. $\left|V_0\right|=0$. & (d) -- Final placement. $\left|V_0\right|=n/2.$
\end{tabular}
\caption{
(a) Torus with 900 nodes, 16 clusters and no nodes in the inactive cluster. Nodes are placed at random positions.
(b) Final placement after applying the local greedy strategy in a round robin fashion. All clusters are grouped together, but their shapes are not optimal for a given cluster.
(c) Torus with 900 nodes, 8 active clusters and half of the nodes in the inactive cluster. $V_0$ has $n/2$ nodes. Nodes are placed at random positions.
(d) Final placement after applying the local greedy strategy in a round robin fashion. All clusters are grouped together and almost optimally shaped around their center. (e) Link to the  animation of this simulation (\url{http://www.bgu.ac.il/\texttildelow avin/pmwiki/pmwiki.php?n=Main.Self-AdjustingNetworks}).
}
\label{fig:clusters_sim}
\end{figure*}

\subsection{Bad Local Minima for Ring Graph}\label{app_ring_bad}
We first give a formal model for the clustered requests distribution.
%To formally model a clustered requests distribution 
We assume, for simplicity, that the clustered requests distribution is a mixture of product distributions on disjoint supports. In the case that the product distributions are uniform, this results in the following.
Let $k$ be the number of groups and let $V_1, \ldots, V_k$ be the vertex disjoint node sets of these groups and let $V_0 = V \setminus \bigcup_i V_i$ be the set of inactive nodes. Let $p_i$ for $1 \leq i \leq k$ be the activity level of each group such that $\sum_{i=1}^k p_i = 1$. Given this we have that the probability $\p(u,v)$ for a request between a nodes $u,v \in V$ is:

\begin{align}
\p(u,v) = \left\{
\begin{array}{c l}
p_i / |V_i|^2 & u,v\in V_i \text{ , }i>0\label{eq:uni_clustered_distr}\\
0 & \text{otherwise}
\end{array}
\right. 
\end{align}

For the 1-D ring graph, the performance of the greedy switching strategy is bad. 
First, the simulation results presented in Figure~\ref{fig:start_end_ratio_ring_torus_1}, show the improvement ratio between the EPL of a random starting locations to the EPL of a local minimum reached by the local greedy M-rule as a function of number of clusters. In the ring graph, the improvement is almost negligible and it is constant with the number of clusters. 
Notice, when a greedy switching strategy is applied to a $2$-dimensional torus graph, the improvement ratio is much better and it is growing with the number of clusters. 
%This is a result of the higher connectivity of the $2$-dimensional torus vs. the $1$-dimensional ring. This leads to the conjecture that this approach might work well in highly connected networks. In the torus, the improvement also depend on the size of the network and for many small clusters it should grow as $O(\sqrt{n})$.

Next, we prove the following lemma that shows that the ratio between the EPL of a worst case local minima and the EPL of the best placement is not a constant anymore but is growing with number of clusters.

\begin{lemma}\label{lemma:bad_local_ring}
In a ring network with uniform clustered distribution of requests, 
the ratio between the worst case stable nodes placement (i.e., a local minima) and the best case placement is at least $1.5c$, where $c$ is the number of clusters.
\end{lemma}

\begin{proof}
First we show that there exists a bad local minimum on the ring graph and uniform clustered distribution of requests. We assume that there are $c$ clusters of $2k+1$ nodes each, and thus the total number of nodes in a ring is $n=c(2k+1)$. Let $\varphi_l$ be the following placement of the nodes: $(v_1^1,v_1^2,v_1^3,\ldots,v_2^1,v_2^2,v_2^3,\ldots)$, where $v_i^j\in V_j$.
We now prove that
the placement of nodes $\varphi_l$ is locally optimal, i.e., there are no improvement switches.

Let's denote the expected cost of the paths where the node $v$ is either source or destination as $\text{E}_v$. Due to the symmetric location and the uniform distribution, $\text{E}_v$ is the same for all the nodes $v\in V$.  
Let's assume that $v\in V_1$, where $V_1$ is the set of nodes in cluster $1$, and $\left| V_1\right |=2k+1$. Let's denote the probability of choosing the node $u\in{V_1}$ as a pair for the given node $v\in{V_1}$ as $p_{u|v}$. Clearly, $p_{u|v}=\frac{1}{|V_1|}$. 
Before the switch:
\begin{align*}
\E_v&=\sum_{u\in V_1}p_{u|v}d_{vu}= \frac{1}{\left| V_1\right |}\cdot 2\sum_{i=1}^k c\cdot i =\frac{c\cdot k(k+1)}{2k+1} 
\end{align*}
Due to the symmetry, there is no difference with what node we will try to switch the node $v$. So, let's assume that we switch it with its right-hand neighbor -- $w$. In order to check the difference with the path cost before the switch, we need to consider all the paths in which nodes $v$ and $w$ are involved. Clearly, even after the switch $\text{E}_v=\text{E}_w$ due to the symmetry, so we will now find the $\text{E}_v$ after the switch and will compare it to the $\text{E}_v$ before the switch.
After the switch:
\begin{align*}
\E_v^{'}&=\sum_{u\in V_1}p_{u|v}d_{vu}\\
&= \frac{1}{\left| V_1\right |}\left(\sum_{i=1}^k (c\cdot i-1)+\sum_{i=1}^k (c\cdot i+1)\right)=\frac{c\cdot k(k+1)}{2k+1} 
\end{align*}
Since all possible switches are the same due to the symmetry, we can conclude that there is no switch that will improve the expected paths cost and the placement $\varphi_l$ is a local minima.

Now let's assume different placement of the nodes -- $\varphi_{opt}$:  $(v_1^1,v_2^1,...,v_{\left | V_1 \right |}^1,v_1^2,v_2^2,...,v_{\left | V_2 \right |}^2,\ldots)$, where $v_i^j\in V_j$.
For this, probably optimal, placement, we will calculate now the expected path length (EPL).

We consider the part of the ring of $2k+1$ nodes that belongs to a single cluster. Number of paths of length 1 inside the cluster is $2\cdot 2k$. Number of paths of length 2 is $2\cdot (2k-1)$, length 3 is $2\cdot (2k-2)$, length $m$ is $2\cdot(2k-m+1)$, and of length $2k$ is $2$. So, we obtain the average path length in one cluster:
\begin{align*}
\E&=\sum_{v,u\in V_1}p_{vu}d_{vu}= \frac{1}{2\left| V_1\right |^2}\cdot2\sum_{i=1}^{2k} m(2k-m+1)\\
&=\frac{k(2k+1)(4k+4)}{3\cdot 2(2k+1)^2}=\frac{k(2k+2)}{3(2k+1)}.
\end{align*}

\end{proof}

\section{Conclusions and Future Work}\label{sec:conclusion}
In this preliminary work, we formally defined the MEPL problem which has practical significance in saving energy of fixed infrastructure network. We showed that in general cases, the problem is hard to compute, but under some realistic assumptions on network infrastructure and traffic patterns, we propose efficient local and distributed algorithms that achieve almost optimal solution. Our algorithms are based on migration of processes, which allows network optimization without changing the underlying infrastructure. This idea integrates well with an SDN concept which will probably include process migration functionality in its management platform.

In future work, we plan to extend our results to other topologies that are used in data centers networks, e.g., fat trees \cite{fat:trees:Leiserson:1985}. We also aim to investigate other types of requests distributions that are based on real data.

\bibliographystyle{IEEEtran}
\bibliography{literature}

\end{document}